\documentclass[glo]{svjour3}

\usepackage[utf8]{inputenc}
\usepackage[T1]{fontenc}
\usepackage[pdftex]{graphicx}
\usepackage{amsmath}
\usepackage{amsfonts}
\usepackage{amssymb}
\usepackage{dsfont}
\usepackage{stmaryrd}
\usepackage{subfigure}
\usepackage{url}
\usepackage{wasysym}
\usepackage{color}
\usepackage{stackrel}
\usepackage{bbm}

\newcommand{\un}[1]{{\underline{#1}}}

\newcommand{\ZZ}{\mathbb{Z}}
\newcommand{\RR}{\mathbb{R}}
\newcommand{\CC}{\mathbb{C}}
\newcommand{\NN}{\mathbb{N}}

\newcommand{\be}{\begin{equation}}
\newcommand{\ee}{\end{equation}}
\newcommand{\bea}{\begin{eqnarray}}
\newcommand{\eea}{\end{eqnarray}}
\def\half{\frac{1}{2}}

\def\tr{{\,{\rm tr}\,}}

\def\one{\mathbbm{1}}

\journalname{Communications in Mathematical Physics}

\begin{document}

\title{Integrability of a deterministic cellular automaton driven by stochastic boundaries}
\titlerunning{Boundary driven integrable cellular automaton}
 
\author{Toma\v{z} Prosen \and Carlos Mej\'ia-Monasterio}

\institute{Toma\v{z} Prosen  \at 
Faculty of Mathematics and Physics, University of Ljubljana,\\
Jadranska 19, SI-1000 Ljubljana, Slovenia\\
\email{tomaz.prosen@fmf.uni-lj.si}
\and
Carlos Mej\'ia-Monasterio \at
Laboratory of Physical Properties, Technical University of Madrid, \\
Av. Complutense s/n, 28040 Madrid, Spain
}

\date{\today}

\maketitle

\begin{abstract}
We propose an interacting many-body space-time-discrete Markov chain model, which is composed of an integrable deterministic and reversible cellular automaton 
(the rule 54 of [Boben\-ko {\em et al}, CMP {\bf 158}, 127 (1993)]) on a finite one-dimensional lattice $(\ZZ_2)^{\times n}$, and local stochastic Markov chains at the two lattice boundaries which provide chemical baths for absorbing or emitting the solitons. Ergodicity and mixing of this many-body Markov chain is proven for generic values of bath parameters,
implying existence of a unique {\em non-equilibrium steady state}. The latter is constructed exactly and explicitly in terms of a particularly simple form of matrix product ansatz which is termed a {\em patch ansatz}. This gives rise to an explicit computation of observables and $k$-point correlations in the steady state as well as the 
construction of a nontrivial set of local conservation laws. Feasibility of an exact solution for the full spectrum and eigenvectors (decay modes) of the Markov matrix is suggested as well. 
We conjecture that our ideas can pave the road towards a theory of integrability of boundary driven classical deterministic lattice systems.
\end{abstract}
 
\maketitle

\section{Introduction}

In 1993, Bobenko {\em et al.} \cite{bob} presented and discussed a simple model of a two-state ($\ZZ_2$) {\em reversible cellular automaton} -- the so-called `rule 54' --   (RCA54) which possesses the main features of integrability and can perhaps be considered as the simplest strongly interacting classical dynamical system with nontrivially scattering soliton solutions. The rule can be encoded via a deterministic local mapping on a diamond-shaped plaquette, $\chi : \ZZ_2\times \ZZ_2\times \ZZ_2 \to \ZZ_2$, determining the state of the {\em south} edge in terms of the states of the  {\em west},  {\em east}, and {\em north} edges
\be
 s_{\rm S} = \chi({s_{\rm W},s_{\rm N},s_{\rm E}}) = s_{\rm N} + s_{\rm W} + s_{\rm E} +  s_{\rm W} s_{\rm E} \pmod{2}, \quad s_{\rm S},s_{\rm N},s_{\rm W},s_{\rm E} \in \ZZ_2,
 \label{chi}
\ee
where the time runs in {\em north -- south} direction. The RCA54 is clearly reversible, as at the same time $s_{\rm N} = \chi(s_{\rm W},s_{\rm E},s_{\rm S})$.
See the bulk (central) part of Fig.~\ref{MCsnap} to observe a typical evolution pattern of the RCA54 rule when applied to some initial configuration in the upper rows.

 
Deterministic  long-time dynamics generated by (\ref{chi}) via permutations over the set ${\cal C} = (\ZZ_2)^{\times n}$  of all possible $2^n$ configurations of the zig-zag lattice (chain) of $n$ subsequent cells (see e.g. Fig.~\ref{Scheme} for schematic illustration) is always, either rather trivial, or depending crucially on the boundary conditions, i.e. the states of the cells at spatial coordinates $x=1$ and $x=n$ which cannot be determined dynamically unless additional rules are introduced. In this paper we propose stochastic update rules for the boundary cells implemented through a pair of local Markov chains which can be interpreted as chemical reservoirs parametrized with absorbtion and emission rates of solitons at each boundary. Note that such a hybrid bulk-deterministic, boundary-stochastic statistical mechanics paradigm is fundamentally different from well-studied boundary driven-diffusive classical simple exclusion processes \cite{derrida}, which have stochastic bulk dynamics but can be interpreted as many-body Markov chains over an identical state space $\RR^{\cal C} \simeq (\RR^2)^{\otimes n}$. Rather, the novel paradigm proposed here should be understood as the simplest classical-dynamical version of integrable boundary driven nonequilibrium quantum spin chains which have recently been intensely studied (see e.g. Ref.~\cite{review} for a review).

After introducing the Markov chain model for boundary driven RCA54 dynamics in section \ref{sect:model}, we shall present the proof of irreducibility and aperiodicity of the resulting Markov matrix, which implies uniqueness of the {\em non-equilibrium steady state} (NESS) and asymptotic approach to NESS from any initial state (ergodicity and dynamical mixing -- sometimes referred to as {\em strong ergodicity} \cite{markov_chains}). The main result of our paper, presented in section \ref{sect:patch}, 
is an exact and explicit solution for NESS in terms of a particular, {\em commutative} but {\em correlated} (non-separable) matrix product ansatz, which we term a {\em patch state ansatz}.
In subsequent section \ref{sect:observables} we shall demonstrate explicit computation of observables, such as steady state values of density, soliton-currents, and arbitrary 
$k-$point spatial density-density correlation functions. Moreover, we demonstrate in section \ref{sect:cl}, that similarly to the case of boundary driven quantum $XXZ$ spin-1/2 chains \cite{prl,review}, one can 
exploit the analytical form of NESS to generate nontrivial (quasi)local conservation laws. In the last section \ref{sect:discussion} we discuss some interesting follow-up questions, such as computation of decay modes and full spectrum of the Markov matrix, and conclude.

\section{Bulk--deterministic, boundary--stochastic Markov--chain soliton model}
\label{sect:model}

Throughout our work we shall -- for simplicity and symmetry reasons -- assume that the number of cells is {\em even}
\be n=2m.\ee 
The extension of the results to the case of {\em odd} sizes $n$ should be straightforward.
We identify a state space with a vector space ${\cal S} = \RR^{\cal C} = (\RR^2)^{\otimes n}$, a linear space embedding a convex subspace of all probability state vectors 
$\un{p}=(p_0,p_1,\ldots,p_{2^n-1}) \in {\cal S}$, satisfying $p_s \ge 0, \sum_{s=0}^{2^n-1} p_s = 1$.
A vector of $k$ binary digits $\un{s} = (s_1,s_2,\ldots,s_k) \in \ZZ_2^{\times k}$ shall often be identified with an integer $s=\sum_{j=1}^k s_k 2^{k-j}$, or, components of the probability state vector 
shall be written as $p_s \equiv p_{(s_1,s_2,\dots,s_n)}\equiv p_{s_1,s_2,\ldots s_n}$. Deterministic, local RCA54 rule in the bulk (\ref{chi}) can be encoded into a $2^3 \times 2^3$ permutation matrix
\begin{equation}
P_{(s,s',s''),(t,t',t'')} = \delta_{s,t} \delta_{s',\chi(t,t',t'')} \delta_{s'',t''},
\end{equation}
or
$$
P = \begin{pmatrix} 
\; 1\; & & & & & & & \cr
& & & \; 1\; & & & & \cr
& & \; 1\; & & & & & \cr
& \;1\; & & & & & & \cr
& & & & & & \;1\; & \cr
& & & & & & & \;1\; \cr
& & & & \;1\; & & & \cr
& & & & & \;1\; & & \end{pmatrix},
$$
which is self invertible, $P^2 = \one_{2^3}$. Here, $\one_d$ denotes a $d-$dimensional identity matrix and $\delta_{s,t}$ a Kronecker symbol.

On the boundaries, we define 2-site local stochastic Markov chains, which depend on the state of a pair of near boundary cells.
Firstly, we define a simple, single-cell (ultralocal) Markov chain, depending on in-flux probability $\alpha$ and out-flux probability $\beta$:
\begin{equation}
E^{\alpha,\beta} = 
\begin{pmatrix} 1-\alpha & \beta \cr \alpha & 1-\beta \end{pmatrix}, \quad \alpha,\beta\in [0,1].
\end{equation}
Secondly, we define 2-cell local Markov chains for a pair of cells near the boundary by imagining another, stochastic cell just beyond the boundary following a Bernoulli process $B(\half,\half)$ and then applying the local RCA54 rule (\ref{chi}) to the triple of cells. Such processes are generated by the following $4\times 4$ Markov matrices matrices, for each boundary
\be
\tilde{P}^{\rm L}_{(s',s''),(t',t'')} = \frac{1}{2} \sum_{s=0}^{1} P_{(s,s',s''),(s,t',t'')}, \quad
\tilde{P}^{\rm R}_{(s,s'),(t,t')} = \frac{1}{2} \sum_{s''=0}^{1} P_{(s,s',s''),(t,t',s'')}.
\ee
These relations can be compactly written in terms of a partial trace $\tr_{\!k}$ over $k$-th qubit of $(\CC^2)^{\otimes 3}$,
$$\tilde{P}^{\rm L} = \half \tr_{\!1} P,\quad 
\tilde{P}^{\rm R} = \half\tr_{\!3} P.$$ 

Composing these two Markov processes, we obtain, for each boundary, the final forms of $4\times 4$ matrices of 2-cell boundary Markov chains 
\bea
&& P^{\rm L} = \tilde{P}^{\rm L} (E^{\alpha,\beta}\otimes \one_2) = 
\begin{pmatrix}
\frac{1}{2} & & \frac{1}{2} & \cr
& 1-\alpha & & \beta \cr
\frac{1}{2} & & \frac{1}{2} & \cr
& \alpha & & 1-\beta \end{pmatrix}, \nonumber \\
&& P^{\rm R} = \tilde{P}^{\rm R} (\one_2 \otimes E^{\gamma,\delta}) = 
\begin{pmatrix}
\frac{1}{2} & \frac{1}{2} & & \cr
\frac{1}{2} & \frac{1}{2} & & \cr
& & 1-\gamma & \delta \cr
& & \gamma & 1-\delta \end{pmatrix}.
\eea
\begin{figure}
 \centering	
\vspace{-1mm}
\includegraphics[width=0.4\columnwidth]{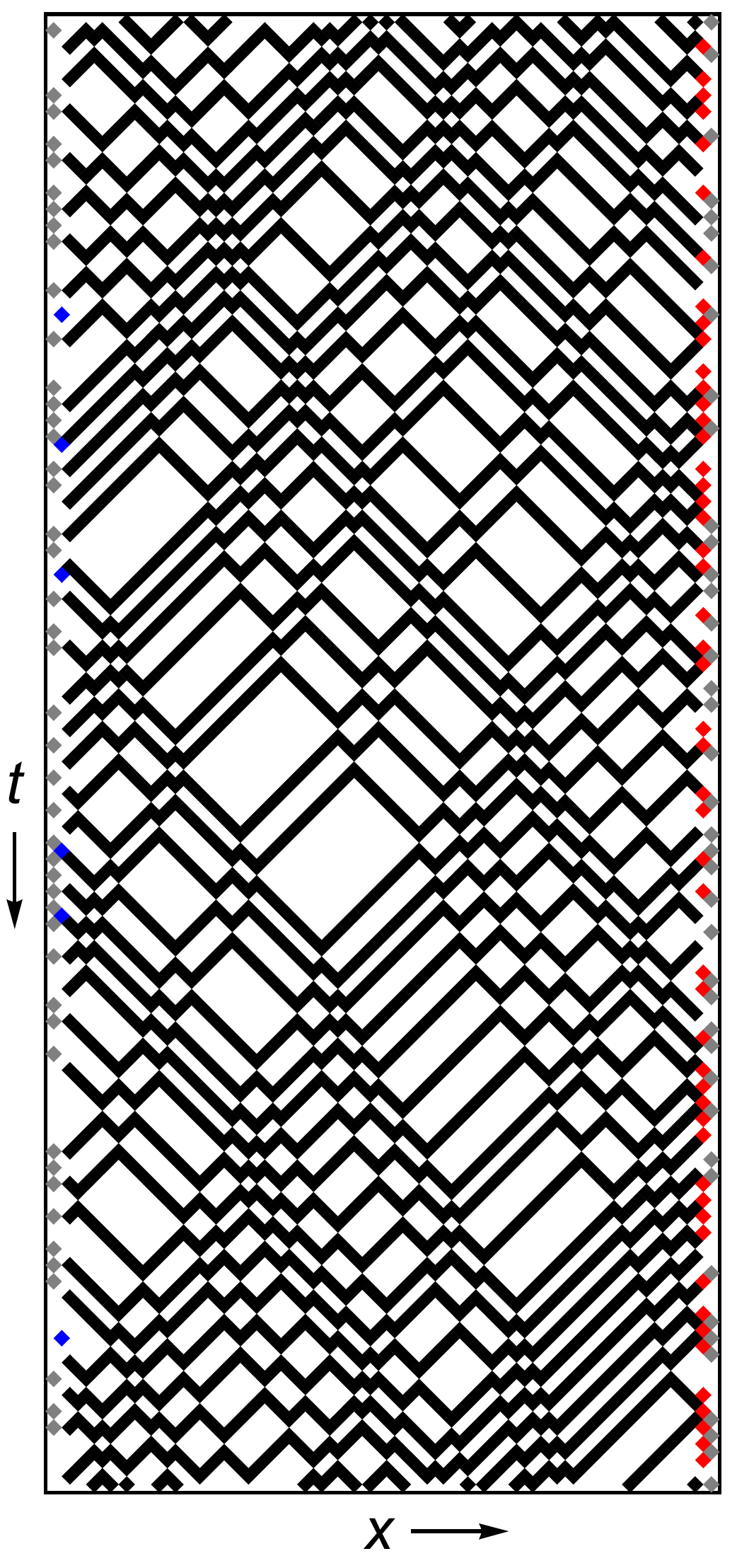}
\vspace{-1mm}
\caption{Monte Carlo dynamics of non-equilibrium stochastically boundary driven deterministic CA, rule 54, for $n=80$, $\alpha=0.1,\beta=0.9,\gamma=0.6,\delta=0.4$.
Time runs downwards. Grey squares denote occupied environmental cells which are generated by a Bernoulli shift with probability $1/2$, while blue (red) squares are occupied boundary cells determined via ultralocal Markov chain $E^{\alpha,\beta}$ ($E^{\gamma,\delta}$). Note that the right end is ``hotter'' than the left one and that the average (steady-state) soliton current points to the left, $J < 0$.}
\label{MCsnap}
\end{figure}

The full Markov chain propagator $U \in {\rm End}({\cal S})$ is then written as a composition of two temporal layer propagators
\begin{eqnarray} 
U &=& U_{\rm o} U_{\rm e}, \label{eq:U}\\
U_{\rm e} &=& P_{123} P_{345} \cdots P_{n-3,n-2,n-1} P^{\rm R}_{n-1,n}, \label{eq:Ue}\\
U_{\rm o} &=& P_{n-2,n-1,n} \cdots P_{456}P_{234}P^{\rm L}_{12} . \label{eq:Uo}
\end{eqnarray}
where embedding into ${\rm End}({\cal S})$ is understood as $P_{k,k+1,k+2} = \one_{2^{k-1}}\otimes P \otimes \one_{2^{n-k-2}}$,
$P^{\rm L}_{1,2} = P^{\rm L} \otimes \one_{2^{n-2}}$, $P^{\rm R}_{n-1,n} = \one_{2^{n-2}}\otimes P^{\rm R}$. The model depends on four external driving parameters $\alpha,\beta,\gamma,\delta \in [0,1]$, which can be understood (or related to) injection/absorption rates of solitons at the left/right boundary, respectively. Monte Carlo dynamics of driven RCA54 for some typical values of driving parameters is illustrated in Fig.~\ref{MCsnap}, while schematic composition of the full many-body Markov generator is depicted in Fig.~\ref{Scheme}

The many-body propagator $U$ is clearly a {\em stochastic matrix}, i.e., its nonnegative elements in each column sum to $1$. 
In fact, for generic values of driving parameters $0 < \alpha,\beta,\gamma,\delta < 1$, it has exactly $4$ nonvanishing matrix elements in each column, corresponding to four combinations of boundary cells $x=1$ and $x=n$. Our goal is to find a steady state probability state vector 
$\un{p}\in {\cal S}$, which is nothing but a fixed point of our many-body Markov chain -- the NESS:
\be
U \un{p} = \un{p},
\label{eq:fp1}
\ee
or equivalently, to find a pair of probability state vectors $\un{p},\un{p}'\in{\cal S}$ on subsequent temporal zig-zag layers, satisfying
\be
U_{\rm e} \un{p} = \un{p'},\quad U_{\rm o}\un{p'} = \un{p}.
\label{eq:fp}
\ee

\begin{figure}
 \centering	
\vspace{-1mm}
\includegraphics[width=0.45\columnwidth]{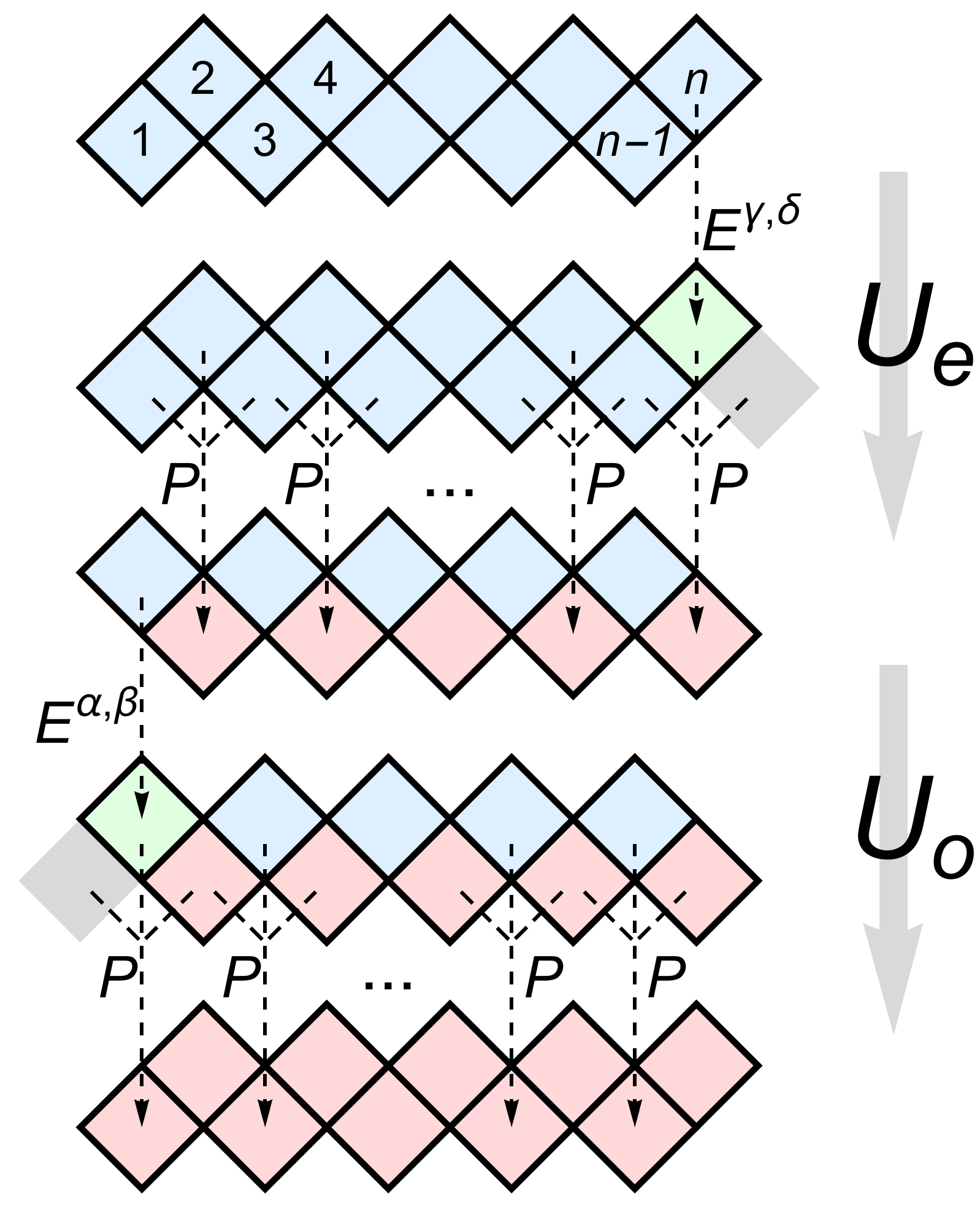}
\vspace{-1mm}
\caption{Schematic illustration of the composition of the full step many-body Markov propagator for the stochastically boundary driven deterministic cellular automaton specified by a local 3-point rule encoded in the permutation matrix $P$. 
The blue cells indicate initial values, while the red cells indicate final values. The green cells are resulting from ultralocal Markov chains $E^{\alpha,\beta}$, while light grey cells indicate probabilistic environment cells which are on/off with probability $1/2$.
}
\label{Scheme}
\end{figure}

\begin{figure}
 \centering	
\vspace{-1mm}
\includegraphics[width=0.26\columnwidth]{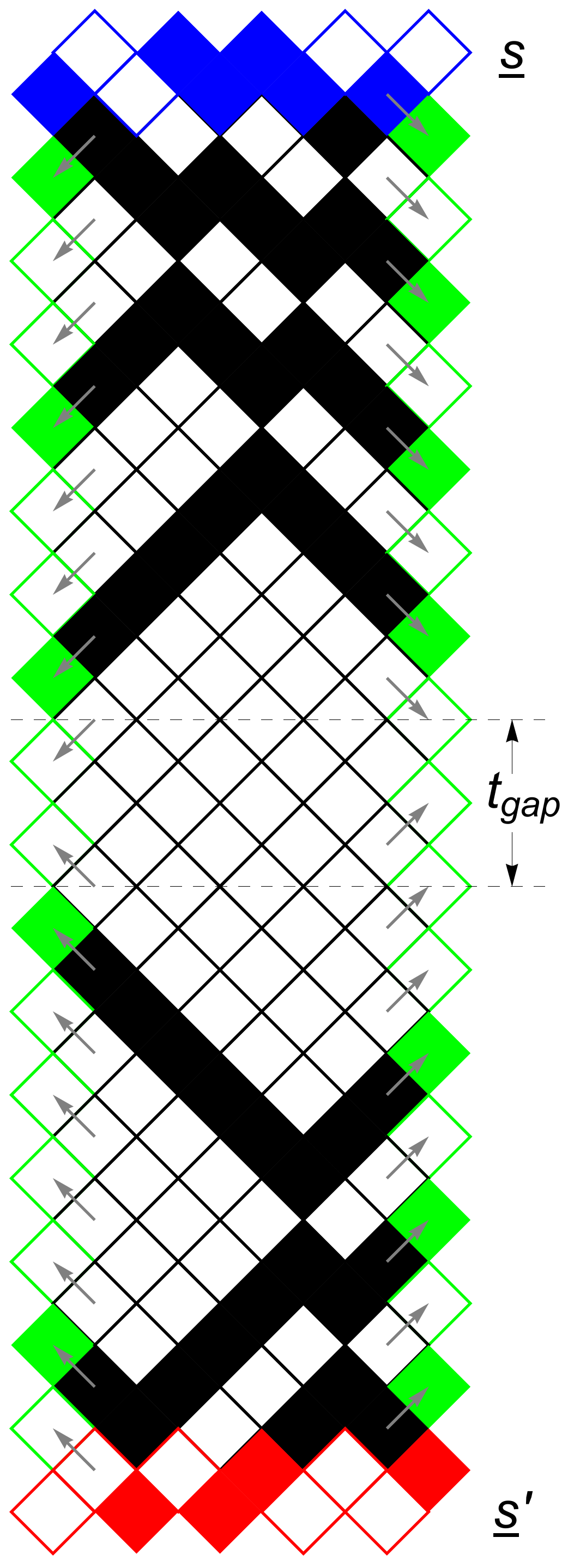}
\vspace{-1mm}
\caption{Illustration of the proof of irreducibility and aperiodicity of the Markov matrix $U$. Blue and red configurations, 
$\un{s}=(1,0,0,1,1,1,1,0,1,0)$ and $\un{s}'=(0,0,1,0,1,1,0,0,0,1)$, are connected with the Markov-graph walk for generic probabilities 
$0<\alpha,\beta,\gamma,\delta < 1$ in at least $t_0=15$ time steps where the boundary cells are chosen as indicated by green cells (the boundary conditions are generated by causal/anti-causal absorbing boundaries in the upper/lower part of the walk). The values of the boundary cells are thus determined by copying the values of the near-by bulk cells in 
the direction of the grey arrows. Consequently $(U^{t_0+t_{\rm gap}})_{\un{s},\un{s}'} > 0$ for any $t_{\rm gap} \ge 0$ 
(and any other pair of initial/final configurations $\un{s},\un{s}'$ with possibly different $t_0$), which implies irreducibility and aperiodicity of $U$).
}
\label{IrredChart}
\end{figure}

However, establishing an existence of a unique NESS and relaxation towards NESS from an arbitrary initial probability state vector amounts to \cite{markov_chains} showing the following statement:

\begin{theorem}
The $2^n \times 2^n$ matrix $U$, Eq.~(\ref{eq:U}), is irreducible and aperiodic for generic values of driving parameters, more precisely, for an open set $0 < \alpha,\beta,\gamma,\delta < 1$.
\end{theorem}

\begin{proof}
We recall \cite{perron_frobenius}  that  a finite, non-negative matrix $U$ is {\em irreducible} if for any pair of configurations $\un{s},\un{s}'\in{\cal C}$,
one can find a natural number $t_0 \in \NN$ such that $(U^{t_0})_{\un{s}',\un{s}} > 0$. An irreducible matrix $U$ is 
{\em aperiodic} if for some configuration $\un{s}\in{\cal C}$, the greatest common divisor of recurrence times $t\in\NN$, for which 
$(U^t)_{\un{s},\un{s}} > 0$, is 1.

Let us first show aperiodicity. As we have argued above, $U_{\un{s}',\un{s}}$ connects each configuration $\un{s}$ with exactly 4 other configurations 
$\un{s}'$ with all possible values of boundary cells, $(s'_1,s'_n)\in\{(0,0),(0,1),(1,0),(1,1)\} $, unless some of the parameters $\alpha,\beta,\gamma,\delta$ is equal exactly 0 or 1 which is the marginal case that is excluded from the discussion. Let us now take a sufficiently large positive integer $t_0$, to be determined below, and fix 
$\un{s}(t_0)\equiv \un{s}',\un{s}(0)\equiv\un{s}$. We shall then construct a walk 
\be
\un{s}(0)\to\un{s}(1)\to\un{s}(2)\to \cdots \to\un{s}(t_0),
\ee 
i.e., a path through the Markov graph defined by positive elements of $U$, which connects $\un{s}$ and $\un{s}'$ in $t_0$ steps and implies $(U^{t_0})_{\un{s}',\un{s}} > 0$ (see Fig.~\ref{IrredChart} for a `self-contained' graphic illustration of the idea of proof). Since we are still free to choose the values of the boundary cells $s_{1,n}(t)$ along the walk $t \in\{1,2,\ldots,t_0-1\}$ apart from the ends.
For the first part of the walk $t = 1,2\ldots t_+$, up to some $t_+<t_0$, we are fixing them with the rule
\be
s_{1}(t) = s_{2}(t-1),\quad s_{n}(t) = s_{n-1}(t-1).
\label{eq:absorb}
\ee
The evolution of the interior values of the cells $s_x(t)$ for $1 < x < n$ and $t \le t_+$ is then completely specified by the deterministic RCA54, while (\ref{eq:absorb}) provide the {\em causal absorbing} boundary conditions. Indeed, each time the boundary cell, say $x=1$, gets occupied, $s_{1}(t)=1$, the soliton as defined in \cite{bob} is absorbed (see Fig.~\ref{IrredChart}). 
As the solitons only move ballistically (at speed 1) and scatter pairwise (with time-lag 1), while they cannot form bound states, it is clear that a finite time scale 
$t_+ \in\NN$ exists, surely smaller than $n^2$, after which all the solitons will be absorbed and we end up in a vacuum configuration $\un{s}(t_+) = (0,0\ldots,0)$. 

For the rest of the walk $t\in\{t_++1,\ldots t_0\}$ we need to show that an alternative boundary rules exist which create the configuration $\un{s}'$ out of the vacuum in another $t_-=t_0-t_+$ steps. This is easily achieved by using {\em time-reversibility} of RCA54 and arguing that a vacuum configuration is again generated from $\un{s}'$ in some $t_-$ steps if the {\em anti-causal absorbing} boundary conditions are set (which are equivalent to (\ref{eq:absorb}) when the time runs backwards)
\be
s_{1}(t) = s_{2}(t+1),\quad s_{n}(t) = s_{n-1}(t+1),\quad {\rm for}\quad t = t_0-1,t_0-2, \ldots, t_0 - t_-\,.
\label{eq:absorb2}
\ee
The entire walk then connects $\un{s}$ to $\un{s}'$ in $t_0=t_++t_-$ steps and implies $(U^{t_0})_{\un{s}',\un{s}} > 0$, for arbitrary pair $\un{s},\un{s}' \in {\cal C}$ where the minimal possible integer $t_0$ may depend on the choice of $\un{s},\un{s}'$. This proves irreducibility of (\ref{eq:U}).

Considering $\un{s}'=\un{s}$, we have just shown that $U^{t_0}_{\un{s},\un{s}} > 0$ for some $t_0$ depending on $\un{s}$.
But since in between annihilating the configuration $\un{s}$ in $t_+$ time steps and then creating it again in another 
$t_-$ steps\footnote{Note that in general $t_-\neq t_+$ as a generic configuration $\un{s}$ is not time-reversal invariant.}, while $t_0=t_++t_-$,
we can await in the vacuum state for an arbitrary additional number of steps $t_{\rm gap} \ge 0$, i.e. increase the walk by a segment of $t_{\rm gap}$ intermediate vacuum configurations, and still have $(U^{t_0 +t_{\rm gap}})_{\un{s},\un{s}} > 0$. The greatest common divisor of the set $\{t_0+t_{\rm gap};t_{\rm gap}\in\ZZ_+\}$ is clearly 1, so we have shown aperiodicity.\hfill$\square$
\end{proof}

In fact, a careful combinatorics of soliton scatterings and boundary absorbtions reveals that the minimal time $t_0$ which suffices for all pairs 
of configurations $\un{s},\un{s}'$, i.e. after which $U^{t_0}$ becomes a (strictly) positive matrix, reads
\be
\min\{ t_0 \in \NN;  (U^{t_0})_{\un{s},\un{s}'} > 0, \forall \un{s},\un{s}' \} = \frac{3}{2} n - 2. 
\ee

In conclusion, the Perron-Frobenius theorem \cite{perron_frobenius} guaranties that NESS probability state vector $\un{p}$ satisfying the fixed point condition (\ref{eq:fp1}), or (\ref{eq:fp}), is {\em unique} -- eigenvalue 1 of $U$ is simple -- and all other eigenvalues of $U$ lie strictly inside the unit circle.
As a consequence, the Markov dynamics $\un{p}(t) = U^t \un{p}(0)$ is {\em ergodic and mixing} and an arbitrary initial probability state vector $\un{p}(0)$ converges exponentially in $t$ to NESS. 

\section{Exact solution of NESS and the patch state ansatz}

\label{sect:patch}

We shall now explicitly construct the probability state vectors $\un{p}$ and $\un{p}'$ of NESS, solving Eq. (\ref{eq:fp}) in terms of a simple ansatz,  which we term a {\em patch state ansatz} (PSA) [illustrated in Fig.~\ref{fig:PSAscheme}].

\begin{theorem}
For an open set of driving parameters, $0 < \alpha,\beta,\gamma,\delta < 1$, the NESS solution $\un{p},\un{p}'\in{\cal S}$ of the fixed point condition (\ref{eq:fp}) can be written, for any even size $n$, in the form
\bea
&& p_{s_1,s_2,\ldots,s_n} = L_{s_1 s_2 s_3} X_{s_2 s_3 s_4 s_5} X_{s_4 s_5 s_6 s_7}  \cdots X_{s_{n-4} s_{n-3} s_{n-2} s_{n-1}} R_{s_{n-2} s_{n-1} s_n}, \nonumber \\
&& p'_{s_1,s_2,\ldots,s_n} = L'_{s_1 s_2 s_3} X'_{s_2 s_3 s_4 s_5} X'_{s_4 s_5 s_6 s_7} \cdots X'_{s_{n-4} s_{n-3} s_{n-2} s_{n-1}} R'_{s_{n-2} s_{n-1}s_n},
\quad
\label{eq:patch}
\eea 
for some rank-4 and rank-3 tensors of strictly positive components $X_{ss'uu'}$, $X'_{ss'uu'}$, $L_{suu'}$, $L'_{suu'}$, $R_{ss'u}$, $R'_{ss'u}$, with binary indices $s,s',u,u'\in\{0,1\}$.

Explicit $n-$independent algebraic expressions for the tensors  $X$, $X'$, $L$, $L'$, $R$, $R'$ in terms of the parameters of the model $\alpha,\beta,\gamma,\delta$ 
shall be given later in the proof (\ref{finalansatz},\ref{solution}).
\end{theorem}
\begin{proof}
We shall first (i) present a minimal set of equations which are sufficient to determine the tensors $X,X',L,L',R,R'$ under the assumption of the theorem. These nonlinear algebraic equations, being sufficiently simple, can be readily solved. Then, in the second part of the proof (ii) we shall show that the PSA solution identically satisfies every  component of the fixed point conditions (\ref{eq:fp}), for any even $n$.

(i) A normalization of the PSA (\ref{eq:patch}) can be chosen such that 
\be
X_{0000}=X'_{0000}=1,\quad L_{000}=R'_{000}=1.
\ee
Clearly $X_{0000}=X'_{0000}$, otherwise the probabilities of the vacuum configurations $p_{0,0,\ldots 0}$ and $p'_{0,0,\ldots,0}$ would scale differently with $n$ which is not possible since $U_{\rm o}$ and $U_{\rm e}$ directly connect $(0,0,\ldots,0,0)$ only with configurations $(s_1,0,\ldots,0,s_n)$.

Let us now assume the ansatz (\ref{eq:patch}) and write all components of  Eqs.~(\ref{eq:fp}),
$(U_{\rm e}\un{p}-\un{p}')_{\un{s}} = (U_{\rm o}\un{p}'-\un{p})_{\un{s}}=0$,
pertaining to 4-cluster configurations in the bulk of the form\footnote{Symbol $0^{\{k\}}$ 
denotes $0$ repeated $k$ times.} $\un{s} = (0^{\{ 2k+1\}},s,s',u,u',0^{\{n-5-2k\}})$, for $k=0,1,\ldots,m-3$, and
3-cluster configurations near each boundary, $\un{s}=(v',s,s',0^{\{n-3\}})$ and $\un{s}=(0^{\{n-3\}},s,s',u)$, resulting in the following finite set of equations\footnote{As a suitable notational convention, {\em primed/unprimed} roman index shall often denote a cell occupation number at {\em odd/even} position.}:
\bea
&&  X'_{0 0 s s'} X'_{s s' u u'} X'_{u u' 0 0} X'_{0000} = \nonumber\\
&& \qquad X_{0 0 \chi(0 s s') s'} X_{\chi(0 s s') s' \chi(s' u u') u'} X_{\chi(s' u u') u' \chi(u' 0 0) 0} X_{\chi(u' 0 0) 000},\nonumber\\
&& X_{0000} X_{0 0 s s'} X_{s s' u u'} X_{u u' 0 0} = \nonumber \\
&& \qquad X'_{000\chi(0 0 s)} X'_{0 \chi(0 0 s) s \chi(s s' u)} X'_{s \chi(s s' u) u \chi(u u' 0)} X'_{u \chi(u u' 0) 0 0},\nonumber\\
&& L'_{v' s s'}X'_{s s' 0 0}X'_{0000}R'_{000} = 
L_{v' \chi(s' s s') s'}X_{\chi(v' s s') s' \chi(s' 0 0) 0}X_{\chi(s' 0 0) 000}\frac{R_{000}+R_{001}}{2}, \nonumber \\
&& L'_{000}X'_{0 0 s s'} R'_{s s' u} = 
L_{000}\sum_{t',t} P^{\rm R}_{(s',u),(t',t)} X_{00 \chi(0 s t') t'} R_{\chi(0 s t') t' t}, \nonumber\\
&&L_{v' s s'} X_{s s' 0 0}R_{000} =
 \sum_{t',t} P^{\rm L}_{(v',s),(t',t)} L'_{t' t \chi(t s' 0)} X'_{t \chi(t s' 0) 0 0}R'_{000},\nonumber\\
&&L_{000}X_{0000}X_{0 0 s s'} R_{s s' u} = 
\frac{L'_{000}+L'_{100}}{2}X'_{000 \chi(0 0 s)}X'_{0 \chi(0 0 s) s \chi(s s' u)} R'_{s \chi(s s' u) u}\, . \label{eqs}
\eea
The total number of $2\times 16+4\times 8 - 4=60$ unknowns can be further reduced by exploring the following {\em gauge symmetry} 
\bea
X_{s s' t t'}  &\longrightarrow& f_{s s'} X_{s s' t t'} f^{-1}_{t t'}, \nonumber \\
L_{s' t t'}  &\longrightarrow& L_{s' t t'} f^{-1}_{t t'},  \nonumber \\
R_{s s' t}  &\longrightarrow& f_{s s'} R_{s s' t}, \nonumber \\
X'_{s s' t t'} & \longrightarrow& g_{s s'} X'_{s s' t t'} g^{-1}_{t t'}, \nonumber \\
L'_{s' t t'} & \longrightarrow& L'_{s' t t'} g^{-1}_{t t'},  \nonumber \\
R'_{s s' t} & \longrightarrow& g_{s s'} R'_{s s' t}, \label{gauge}
\eea
which conserves the patch ansatz (\ref{eq:patch}), as well as the defining equations (\ref{eqs}) for arbitrary nonzero gauge `fields' $f_{s s'}, g_{s s'}$. We can uniquely fix $f_{s s'}, g_{s s'}$ by choosing the following gauge
$X_{00 s s'} = X'_{0 0 s s'} = 1$, $\forall s,s' \in \{ 0, 1\}.$ 
Numerical experiments suggest some further symmetries which finally inspire the following ansatz
\bea\qquad\quad
X_{0000} = 1, \quad & X'_{0000} = 1, \qquad\quad L_{000} = 1,\quad & L'_{000} = 1,\nonumber\\
X_{0001} = 1, \quad & X'_{0001} = 1, \qquad\quad L_{001} = 1,\quad & L'_{001} = x_6,\nonumber\\
X_{0010} = 1, \quad & X'_{0010} = 1, \qquad\quad L_{001} = 1,\quad & L'_{001} = x_6,\nonumber\\
X_{0011} = 1, \quad & X'_{0011} = 1, \qquad\quad L_{010} = x_7,\quad & L'_{010} = x_7,\nonumber\\
X_{0100} = x_1, \quad & X'_{0100} = x_1, \qquad\quad L_{011} = x_6,\quad & L'_{011} = 1,\nonumber\\
X_{0101} = x_1, \quad & X'_{0101} = x_1, \qquad\quad L_{100} = 1,\quad & L'_{100} = x_8,\nonumber\\
X_{0110} = x_2, \quad & X'_{0110} = x_3, \qquad\quad L_{101} = 1,\quad & L'_{101} = x_9,\nonumber\\
X_{0111} = x_4, \quad & X'_{0111} = x_5, \qquad\quad L_{111} = x_{9},\quad & L'_{111} = 1,\nonumber\\
X_{1000} = x_1, \quad & X'_{1000} = x_1, \qquad\quad R_{000} = x_{12},\quad & R'_{000} = 1,\nonumber\\
X_{1001} = x_1, \quad & X'_{1001} = x_1, \qquad\quad R_{001} = x_{13},\quad & R'_{001} = 1,\nonumber\\
X_{1010} = x_1, \quad & X'_{1010} = x_1, \qquad\quad R_{010} = x_{14},\quad & R'_{010} = x_{15},\nonumber\\
X_{1011} = x_1, \quad & X'_{1011} = x_1, \qquad\quad R_{011} = x_{12},\quad & R'_{011} = x_{16},\nonumber\\
X_{1100} = x_5, \quad & X'_{1100} = x_4, \qquad\quad R_{100} = x_{17},\quad & R'_{100} = x_{18},\nonumber\\
X_{1101} = x_5, \quad & X'_{1101} = x_4, \qquad\quad R_{101} = x_{19},\quad & R'_{101} = x_{18},\nonumber\\
X_{1110} = 1, \quad & X'_{1110} = 1, \qquad\quad R_{110} = x_{20},\quad & R'_{110} = x_{21},\nonumber\\
X_{1111} = x_1, \quad & X'_{1111} = x_1, \qquad\quad R_{111} = x_{20},\quad & R'_{111} = x_{22}, \label{finalansatz}
\eea
with $22$ unknown parameters/variables $\{x_i; i=1,\ldots,22\}$.
Assuming that all components are nonvanishing, i.e., $x_j \neq 0$, the defining relations (\ref{eqs}) are equivalent to the following set of polynomial 
equations
\bea
&& x_1 x_2 - x_4 = 0,  \qquad x_3 x_4 - 1 = 0, \qquad x_1 x_3 - x_5 = 0, \qquad x_4 x_5 - x_1 = 0, \nonumber\\
&& x_8 - 2 x_{12}  + 1 = 0, \quad x_6 + x_9 - 2 x_{12} = 0,\quad x_{12}+x_{13}-2=0,\nonumber\\
&& 2 x_1 x_{11} - x_{12} - x_{13} = 0, \;\; 2 x_4 - x_1 x_2 (x_{12} + x_{13}) = 0,\;\; 2 x_8 - x_1 x_{10}(x_{12} + x_{13}) = 0,\nonumber\\
&& 2 x_{14} - (1+ x_8) x_{15} = 0,\quad 2 x_{13} - (1 +x_8) x_{16} =0,\quad x_{17} - 2 x_{18} + x_{19} =0,\nonumber\\
&& 2 x_{20} - x_3 (1+x_8) x_{18} = 0,\;\; 2 x_{19} - x_5 (1+x_8) x_{22}=0,\;\; 2 x_{17} - x_5 (1+ x_8) x_{21} =0, \nonumber\\
&& x_1(\alpha x_7 + (1-\beta) x_{11}) - x_5 x_6 x_{12} = 0, \qquad  (\beta-\alpha-1)x_4 + x_1 x_7 x_{12} = 0,  \nonumber\\
&& x_1((1-\alpha) x_7 + \beta x_{11}) - x_5 x_9 x_{12} = 0,  \qquad   (\alpha-\beta-1)x_4 + x_1 x_{10} x_{12} = 0, \nonumber\\
&& (1-\delta) x_{12} + \gamma x_{14} - x_{21} = 0, \qquad x_{16} + (\gamma -\delta - 1) x_{20} = 0, \nonumber \\
&& \delta x_{12} + (1-\gamma) x_{14} - x_{22} = 0, \qquad x_{15} + (\delta - \gamma - 1) x_{20} = 0 \, . \label{eqs2}
\eea
Luckily, this system of nonlinear equations admits a simple solution, which can be compactly written introducing the {\em difference driving parameters}:
\be
\lambda = \alpha - \beta,\qquad \mu = \gamma - \delta,
\ee
namely:
\bea
    x_1 &=& \frac{(\lambda +2) (\mu +2)}{(\lambda  \mu +\lambda -2) (\lambda  \mu
   +\mu -2)}, \nonumber \\
   x_2&=& -\frac{(\lambda  \mu +\lambda -2)^2}{(\lambda +2) (\lambda  \mu
   +\mu -2)}, \nonumber \\ 
   x_3 &=& -\frac{(\lambda  \mu +\mu -2)^2}{(\mu +2) (\lambda  \mu +\lambda
   -2)}, \nonumber \\ 
   x_4&=& -\frac{(\mu +2) (\lambda  \mu +\lambda -2)}{(\lambda  \mu +\mu
   -2)^2}, \nonumber \\ 
   x_5 &=& -\frac{(\lambda +2) (\lambda  \mu +\mu -2)}{(\lambda  \mu +\lambda
   -2)^2}, \nonumber \\
    x_6 &=& -\frac{(\lambda  \mu +\lambda -2) ((\lambda +1) \mu  (\lambda + 1 - 2\alpha)-2 (\alpha+1) \lambda -2)}{(\lambda +2) (\lambda  \mu +\mu
   -2)}, \nonumber \\ 
   x_7 &=& -\frac{(\lambda +1) (\lambda  \mu +\lambda -2)}{\lambda +2}, \nonumber \\ 
   x_8 &=&
   \frac{(\lambda -1) (\mu +2)}{\lambda  \mu +\mu -2}, \nonumber \\ 
   x_9 &=& \frac{(\lambda  \mu
   +\lambda -2) ((\lambda +1) \mu  (\lambda + 1 - 2\alpha)-2 \alpha \lambda
   +2)}{(\lambda +2) (\lambda  \mu +\mu -2)}, \nonumber \\ 
   x_{10} &=& \frac{(\lambda -1) (\lambda  \mu
   +\lambda -2)}{\lambda +2}, \nonumber \\ 
   x_{11} &=& \frac{(\lambda  \mu +\lambda -2) (\lambda  \mu
   +\mu -2)}{(\lambda +2) (\mu +2)}, \nonumber \\ 
   x_{12} &=& \frac{\lambda  \mu +\lambda -2}{\lambda 
   \mu +\mu -2}, \nonumber \\ 
   x_{13} &=& \frac{(\lambda +2) (\mu -1)}{\lambda  \mu +\mu -2}, \nonumber \\ 
   x_{14} &=&
   -\frac{(\lambda +2) (\mu +1)}{\lambda  \mu +\mu -2}, \nonumber \\ 
   x_{15} &=& -\frac{(\lambda +2)
   (\mu +1)}{\lambda  \mu +\lambda -2}, \nonumber \\ 
   x_{16} &=& \frac{(\lambda +2) (\mu -1)}{\lambda 
   \mu +\lambda -2}, \nonumber \\ 
   x_{17} &=& -\frac{(\lambda +2) (\lambda  (\mu +1) (\mu
   +1 - 2\gamma)-2 (\gamma+1) \mu -2)}{(\lambda  \mu +\lambda -2) (\lambda  \mu +\mu
   -2)}, \nonumber \\ 
   x_{18} &=& \frac{(\lambda +2) (\mu +2)}{(\lambda  \mu +\lambda -2) (\lambda  \mu
   +\mu -2)}, \nonumber \\ 
   x_{19} &=& \frac{(\lambda +2) (\lambda  (\mu +1) (\mu + 1 - 2\gamma)-2 \gamma \mu +2)}{(\lambda  \mu +\lambda -2) (\lambda  \mu +\mu -2)}, \nonumber \\ 
   x_{20} &=&
   -\frac{\lambda +2}{\lambda  \mu +\lambda -2}, \nonumber \\ 
   x_{21} &=& \frac{\lambda  (\mu +1) (\mu +1 - 2\gamma)-2 (\gamma+1) \mu -2}{\lambda  \mu +\mu -2}, \nonumber \\ 
   x_{22} &=& -\frac{\lambda
    (\mu +1) (\mu+1 -2\gamma)-2 \gamma \mu +2}{\lambda  \mu +\mu -2}. \label{solution}
\eea
Note that tensors $X$ and $X'$ (components $x_1,\ldots, x_5$), which determine the bulk properties of NESS, depend {\em only} on the difference parameters $\lambda,\mu$, while some components of the boundary tensors $L,R,L',R'$ (namely $x_6,x_9,x_{17},x_{19},x_{21},x_{22}$)
depend explicitly also on the offset parameters $\alpha,\gamma$.
 Furthermore, {\em all} components are strictly positive, $x_j > 0,  j=1,\ldots,22$, on the open physical domain $(\alpha,\beta,\gamma,\delta)\in (0,1)^{\times 4}$. Notice as well that the Gr\" obner basis algorithm implemented within Mathematica yielded one more solution to the system of Eqs.~(\ref{eqs2}), which however can be excluded by verifying components of the fixed point condition (\ref{eq:fp}) beyond the 3-cluster and 4-cluster configurations.

\begin{figure}
 \centering	
\vspace{-1mm}
\includegraphics[width=0.5\columnwidth]{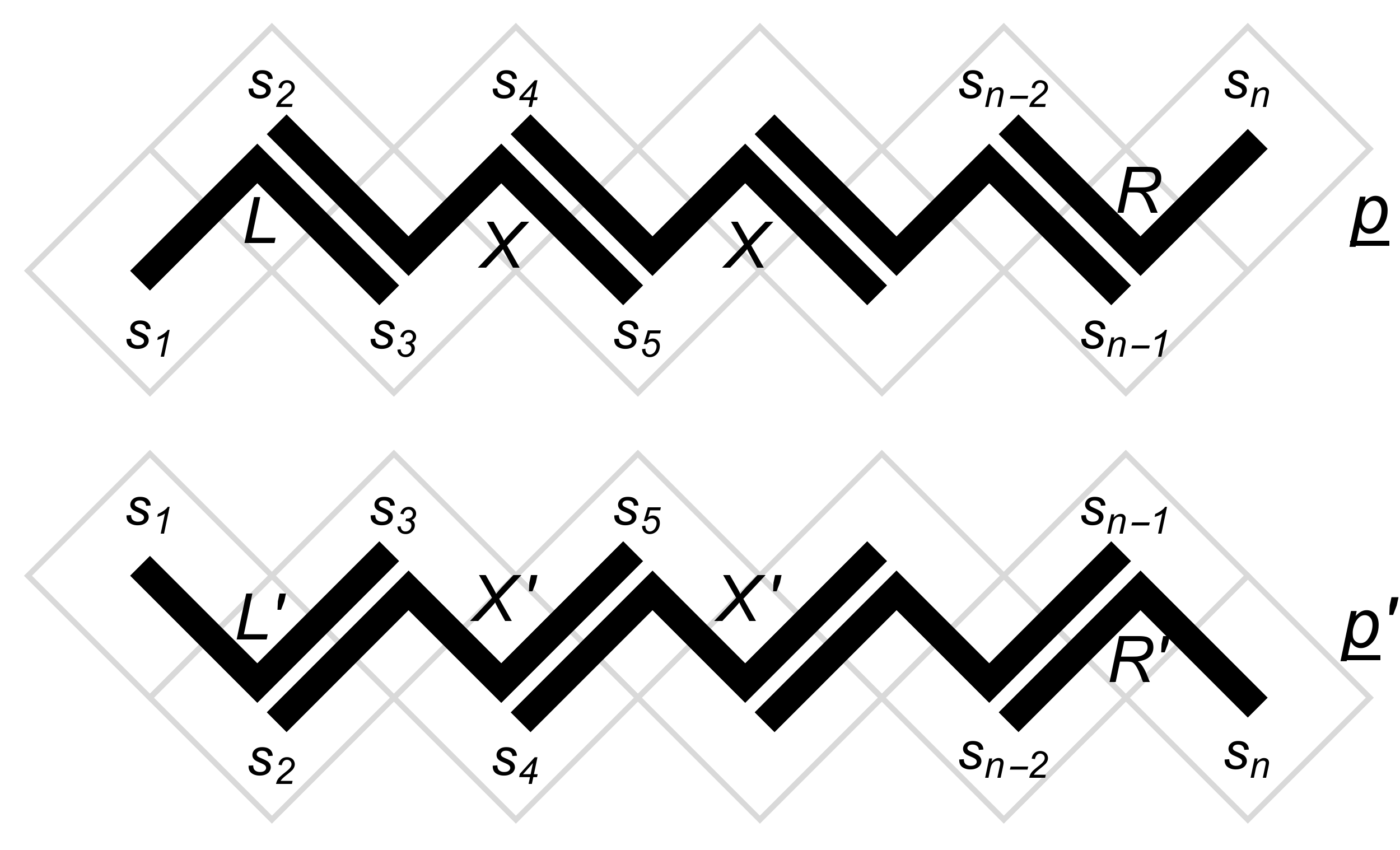}
\vspace{-1mm}
\caption{Illustration of the patch state ansatz (\ref{eq:patch}) for NESS probability state vectors $\un{p},\un{p}'$.}
\label{fig:PSAscheme}
\end{figure}

(ii) Having explicit expressions for the patch tensors $X,X',L,L',R,R'$ (\ref{solution}) one can now check that the ansatz (\ref{eq:patch}) solves the full set of   component-wise equations of the fixed point condition (\ref{eq:fp}) for a small system size, say for $n=8$ where the PSA contains products of two subsequent $X$ ($X'$) tensors. This is readily confirmed using computer algebra. Similarly, one may verify the following set of remarkable composition identities
\bea
\frac{X'_{ss'tt'}X'_{tt'uu'}}{X'_{ss'uu'}} &=& \frac{X_{\chi(v'ss')s'\chi(s'tt')t'}X_{\chi(s'tt')t'\chi(t'uu')u'}X_{\chi(t'uu')u'\chi(u'zz')z'}}{X_{\chi(v'ss')s'\chi(s'uu')u'}X_{\chi(s'uu')u'\chi(u'zz')z'}}, \label{ext1}\\
\frac{X_{ss'tt'}X_{tt'uu'}}{X_{ss'uu'}} &=& \frac{X'_{z\chi(zz's)s\chi(ss't)}X'_{s\chi(ss't)t\chi(tt'u)}X'_{t\chi(tt'u)u\chi(uu'v)}}{X'_{z\chi(zz's)s\chi(ss'u)}X'_{s\chi(ss'u)u\chi(uu'v)}}, \label{ext2}
\eea
for an arbitrary configuration of indices $s,s',t,t',u,u',z,z'$, and $v'$ for (\ref{ext1}) or $v$ for (\ref{ext2}), that means in total $2\times 2^9=1024$ identities.

Now we can make an inductive step. We assume that (\ref{eq:patch}) solves (\ref{eq:fp}) for some even $n$, writing the PSA compactly as
\bea
p_{\ldots v' s s' u u' z z'\ldots} &=&  X^{\rm L}_{\ldots v' s s'} X_{s s' u u'} X_{u u' z z'} X^{\rm R}_{z z' \ldots}, \\
p'_{\ldots z z' s s' u u' v \ldots} &=& X^{\prime \rm L}_{\ldots z z'} X'_{z z' s s'} X'_{s s' u u'} X^{\prime\rm R}_{u u' v \ldots}, 
\eea
where $X^{{\rm L}}_{\un{s}}$, or $X^{\prime{\rm L}}_{\un{s}}$, denote a {\em patch product}\footnote{General component-wise product with two overlapping adjacent indices between each pair of tensor-component factors, exactly like in Eqs.~(\ref{eq:patch}).} of a tensor $L$, or $L'$, with an appropriate number (which could also be zero) of tensors 
$X$, or $X'$, while $X^{{\rm R}}_{\un{s}}$, or $X^{\prime{\rm R}}_{\un{s}}$, denote a patch product of an appropriate number of tensors $X$, or $X'$, with $R$, or $R'$. Then, the condition that the PSA
\bea
p_{\ldots v' s s' t t' u u' z z'\ldots} &=&  X^{\rm L}_{\ldots v' s s'} X_{s s' t t'} X_{t t' u u'} X_{u u' z z'} X^{\rm R}_{z z' \ldots}, \\
p'_{\ldots z z' s s' t t' u u' v \ldots} &=& X^{\prime \rm L}_{\ldots z z'} X'_{z z' s s'} X'_{s s' t t'} X_{t t' u u'} X^{\prime\rm R}_{u u' v \ldots}, 
\eea
for the system of size $n+2$ again solves the fixed point equations (\ref{eq:fp}), namely
\bea
\frac{ (U_{\rm e}\un{p})_{\ldots v' s s' t t' u u' z z'\ldots} }{ (U_{\rm e}\un{p})_{\ldots v' s s' u u' z z'\ldots}} &=& 
\frac{ p'_{\ldots v' s s' t t' u u' z z'\ldots} }{p'_{\ldots v' s s' u u' z z'\ldots}}, \\
\frac{ (U_{\rm o}\un{p}')_{\ldots z z' s s' t t' u u' v \ldots}}{ (U_{\rm o}\un{p}')_{\ldots z z' s s' u u' v \ldots}} &=& 
\frac{ p_{\ldots z z' s s' t t' u u' v \ldots} }{p_{\ldots z z' s s' u u' v \ldots}},
\eea
amounts exactly to $X$-system (\ref{ext1},\ref{ext2}).\hfill$\square$
\end{proof}

We could have even started the induction at $n=6$ as an additional set of $2\times 2^7$ composition identities involving the boundary tensors can be straightforwardly verified using computer algebra as well
\bea
\frac{L'_{s'tt'}X'_{tt'uu'}}{L'_{s'uu'}} &=& \frac{L_{s'\chi(s'tt')t'}X_{\chi(s'tt')t'\chi(t'uu')u'}X_{\chi(t'uu')u'\chi(u'zz')z'}}{L_{s'\chi(s'uu')u'}X_{\chi(s'uu')u'\chi(u'zz')z'}}, \label{ext3}\\
\frac{X_{ss'tt'}R_{tt'u}}{R_{ss'u}} &=& \frac{X'_{z\chi(zz's)s\chi(ss't)}X'_{s\chi(ss't)t\chi(tt'u)}R'_{t\chi(tt'u)u}}{X'_{z\chi(zz's)s\chi(ss'u)}R'_{s\chi(ss'u)u}}. \label{ext4}
\eea

\section{Lax pair, observables and correlations in the steady state}

\label{sect:observables}

Having an explicit result for the NESS probability vectors $\un{p},\un{p}'$ (\ref{eq:patch},\ref{finalansatz},\ref{solution}) at hand we shall now
address a natural question of computation of physical observables, such as density profiles and density-density correlations in the steady state.
To start with, let us define the nonequilibrium partition function via normalization of the total probability of the state, and the corresponding transfer matrix.

We observe the following obvious identities, noting again that $n=2m$,
\be
Z_{n} = \sum_{\un{s}\in{\cal C}} p_{\un{s}} =\un{l} \cdot T^{m-2} \un{r}, \qquad
Z'_{n} = \sum_{\un{s}\in{\cal C}} p'_{\un{s}} = \un{l}' \cdot {T'}^{m-2} \un{r}',
\ee
where 
\bea
\un{l} = \un{l}_0 + \un{l}_1,\quad \un{l}' = \un{l}'_0 + \un{l}'_1, \quad
\un{r} = \un{r}_0 + \un{r}_1,\quad \un{r}' = \un{r}'_0 + \un{r}'_1,
\eea
and
$\un{l}_u,\un{l}'_u,\un{r}_u,\un{r}'_u$, $u\in\{0,1\}$, are vectors from $\RR^4$ with components labelled as $0,1,2,3$, namely:
\bea
&& (\un{l}_{u'})_{2s+s'} = L_{u'ss'},\quad (\un{r}_{u})_{2s+s'} = R_{ss'u}, \nonumber \\
&& (\un{l}'_{u'})_{2s+s'} = L'_{u'ss'},\quad (\un{r}'_{u})_{2s+s'} = R'_{ss'u},
\eea
and $T,T'\in{\rm End}(\RR^4)$ are $4\times 4$ transfer matrices with components
\bea
T_{2s+s',2u+u'} = X_{ss'uu'},\quad T'_{2s+s',2u+u'}=X'_{ss'uu'}\, .
\eea
Straightforward calculation by means of the explicit solution (\ref{solution}) shows that $T$ and $T'$ are similar, i.e. $\exists W\in{\rm End}(\RR^4)$, 
\be
T' = W T W^{-1},
\label{eq:isospc}
\ee 
such that also 
\be
\un{l}' = \kappa \un{l} W^{-1},\quad \un{r}' = \kappa^{-1} W \un{r},
\ee
where
\be
\kappa = \frac{(\lambda +2) (\lambda  \mu +\mu
   -2)}{(\mu +2) (\lambda  \mu +\lambda -2)},
\ee
and consequently
\be
Z_n = Z'_n,
\ee
for any pair of driving parameters $\lambda,\mu \in (-1,1)$, which is nothing but the conservation of total probability.

Writing a pair of independent {\em spectral parameters} as
\be
\omega  \equiv x_4 = \varphi(\lambda,\mu),\;\; \xi \equiv x_5 = \varphi(\mu,\lambda),\quad \varphi(\lambda,\mu)\equiv
-\frac{(\mu +2) (\lambda  \mu +\lambda -2)}{(\lambda  \mu +\mu
   -2)^2}
\ee
we find compact expressions for the transfer- and the intertwining matrices
\bea
&& T = \left(
\begin{array}{cccc}
 1 & 1 & 1 & 1 \\
 \xi  \omega  & \xi  \omega  & \frac{1}{\xi } & \omega  \\
 \xi  \omega  & \xi  \omega  & \xi  \omega  & \xi  \omega 
   \\
 \xi  & \xi  & 1 & \xi  \omega  \\
\end{array}
\right), \quad T' = \left(
\begin{array}{cccc}
 1 & 1 & 1 & 1 \\
 \xi  \omega  & \xi  \omega  & \frac{1}{\omega } & \xi  \\
 \xi  \omega  & \xi  \omega  & \xi  \omega  & \xi  \omega 
   \\
 \omega  & \omega  & 1 & \xi  \omega  \\
\end{array}
\right), \\
&& 
W = \left(
\begin{array}{cccc}
 \frac{\omega+1}{\omega } & \frac{1}{\omega } & -\frac{1}{\xi 
   \omega } & -\frac{1}{\xi  \omega } \\
 -\frac{\omega+1}{\omega } & -\frac{1}{\omega } &
   \frac{\xi +1}{\xi  \omega } & \frac{\xi +1}{\xi  \omega
   } \\
 \xi  & \xi  & 0 & -1 \\
 0 & 0 & 0 & 1 \\
\end{array}
\right).
\eea
Remarkably, $T$ and $T'$ are swapped upon an exchange of driving parameters $\lambda$ and $\mu$, or equivalently, exchanging the spectral parameters 
$\omega$ and $\xi$:
\be
T(\omega,\xi) \equiv T'(\xi,\omega).
\ee
In fact, $W$ acts as an {\em intertwiner} connecting two subsequent temporal layers, 
\be
W(\omega,\xi) T(\omega,\xi) = T(\xi,\omega) W(\omega,\xi),
\label{eq:iso2}
\ee
satisfying an {\em inversion identity}
\be
[W(\omega,\xi)]^{-1} = W(\xi,\omega).
\ee
On the other hand, Eq. (\ref{eq:isospc}) [or (\ref{eq:iso2})] can be interpreted as an {\em isospectral problem} 
(or discrete space-time zero-curvature condition) with a {\em Lax pair}
$\{ T(\omega,\xi), W(\omega,\xi) \}$. This is a clear indication of Lax integrability of our nonequilibrium steady state.
 
The transfer matrices $T,T'$ are singular with, in general, three nonvanishing eigenvalues, 
$T=V {\rm diag}(\tau_1,\tau_2,\tau_3,0) V^{-1}, 
T' = V' {\rm diag}(\tau_1,\tau_2,\tau_3,0) {V'}^{-1}$, with $W=V' V^{-1}$, reading explicitly
\bea
\tau_1 &=& \frac{(\lambda  \mu -4)^2}{(\lambda  \mu +\lambda -2)
   (\lambda  \mu +\mu -2)}, \nonumber \\
\tau_{2,3} &=& \frac{ \lambda \mu (\lambda + \mu + 8) + 4 (\lambda + \mu) \pm (\lambda - \mu)\sqrt{D}}{2 (\lambda  \mu
   +\lambda -2) (\lambda  \mu
   +\mu -2)}, \nonumber \\
D &=&   \lambda \mu (\lambda \mu - 12) - 8 (\lambda + \mu).
\label{tau}
   \eea
We remark two important facts: (i) in the whole open parameter domain $\mu,\lambda \in (-1,1)$, $\tau_1$ represents the leading eigenvalue $\tau_1 > |\tau_{2,3}|$.
(ii) In the limit $\mu,\lambda \to 1$, the three eigenvalues collapse, $\tau_1/\tau_{2,3} \to 1$, so we should see long range correlations there (to be discussed below); namely, the correlation length should diverge as $\lambda,\mu \to 1$. Also, an algebraic curve along which the discriminant vanishes $D(\lambda,\mu)=1$ is potentially interesting, as it signals discontinuous behaviour separating the regime of $|\tau_2|=|\tau_3|$ from the regime of $|\tau_2| \neq |\tau_3|$. 

A key feature of integrability of our boundary driven CA is a {\em compatibility condition} between the bulk transfer matrix and the boundary vectors, namely
\be
\un{l} T = \tau_1 \un{l},\quad \un{l}' T' = \tau_1 \un{l}',\quad T\un{r} = \tau_1 \un{r},\quad T'\un{r}' = \tau_1 \un{r}',
\label{magic}
\ee
relations, which can be readily verified from our analytic solution. This yields a particularly simple expression for the partition sum
\be
Z_n = (\un{l}\cdot\un{r}) \tau_1^{m-2}.
\ee
We are now ready to compute the density profiles. 
Let us define the following pair of diagonal matrices 
\be
D_{\rm e}={\rm diag}(0,0,1,1),\quad  D_{\rm o}={\rm diag}(0,1,0,1).
\ee
Then the steady-state density profiles on even and odd {\em bulk} sites express as:
\bea
\rho_{2k} &=& \frac{1}{Z_n}\sum_{s_1,s_2,\dots s_n} s_{2k}\,p_{s_1,s_2,\dots s_n} = \frac{1}{Z_n} \un{l}\cdot T^{k-1} D_{\rm e} T^{m-1-k} \un{r},  \nonumber \\
\rho_{2k+1} &=& \frac{1}{Z_n}\sum_{s_1,s_2,\dots s_n} s_{2k+1}\,p_{s_1,s_2,\dots s_n} = \frac{1}{Z_n}  \un{l}\cdot T^{k-1} D_{\rm o} T^{m-1-k} \un{r} \, , \label{density}
\eea
for $k=1,2,\ldots m-1$, while at the boundary sites
\be
\rho_{1} = \frac{1}{Z_n} \un{l}_1\cdot T^{m-2} \un{r},\quad
\rho_{n} = \frac{1}{Z_n} \un{l}\cdot T^{m-2} \un{r}_1 \,.
\ee
The compatibility conditions (\ref{magic}) immediately imply flat -- ballistic -- density profiles apart from the boundary sites
\bea
\rho_1 &=& \frac{\un{l}_1\cdot\un{r}}{\un{l}\cdot\un{r}} = \frac{2(1-\alpha)(\lambda\mu+\lambda+\mu)+(\lambda-2)(\lambda\mu-4)}{2
   ( \lambda +\mu +8 - \lambda \mu)},  \nonumber \\
\rho_{j} &=& \frac{\un{l}\cdot D_{\rm e}\un{r}}{\un{l}\cdot\un{r}} =  \frac{\un{l}\cdot D_{\rm o}\un{r}}{\un{l}\cdot\un{r}} = 
  \frac{\lambda +\mu +4}{ \lambda +\mu +8 - \lambda\mu}, \qquad {\rm for}\;\; 1 < j < n,\\
\rho_n &=& \frac{\un{l}\cdot\un{r}_1}{\un{l}\cdot\un{r}}= \frac{2(1-\gamma)(\lambda  \mu +\lambda
   +\mu )+(\mu -2) (\lambda  \mu -4)}{2
   ( \lambda +\mu +8 - \lambda \mu)}\, .\nonumber
\eea
Interestingly, the bulk steady-state density can only take values in the interval $\rho_j \in (\frac{2}{5}, \frac{2}{3})$, with the extreme value of 
maximum density $\frac{2}{3}$ (minimum density $\frac{2}{5}$) reached in the limit $\mu,\lambda\to 1$ ($\mu,\lambda\to -1$).

\begin{figure}
 \centering	
\vspace{-1mm}
\includegraphics[width=0.8\columnwidth]{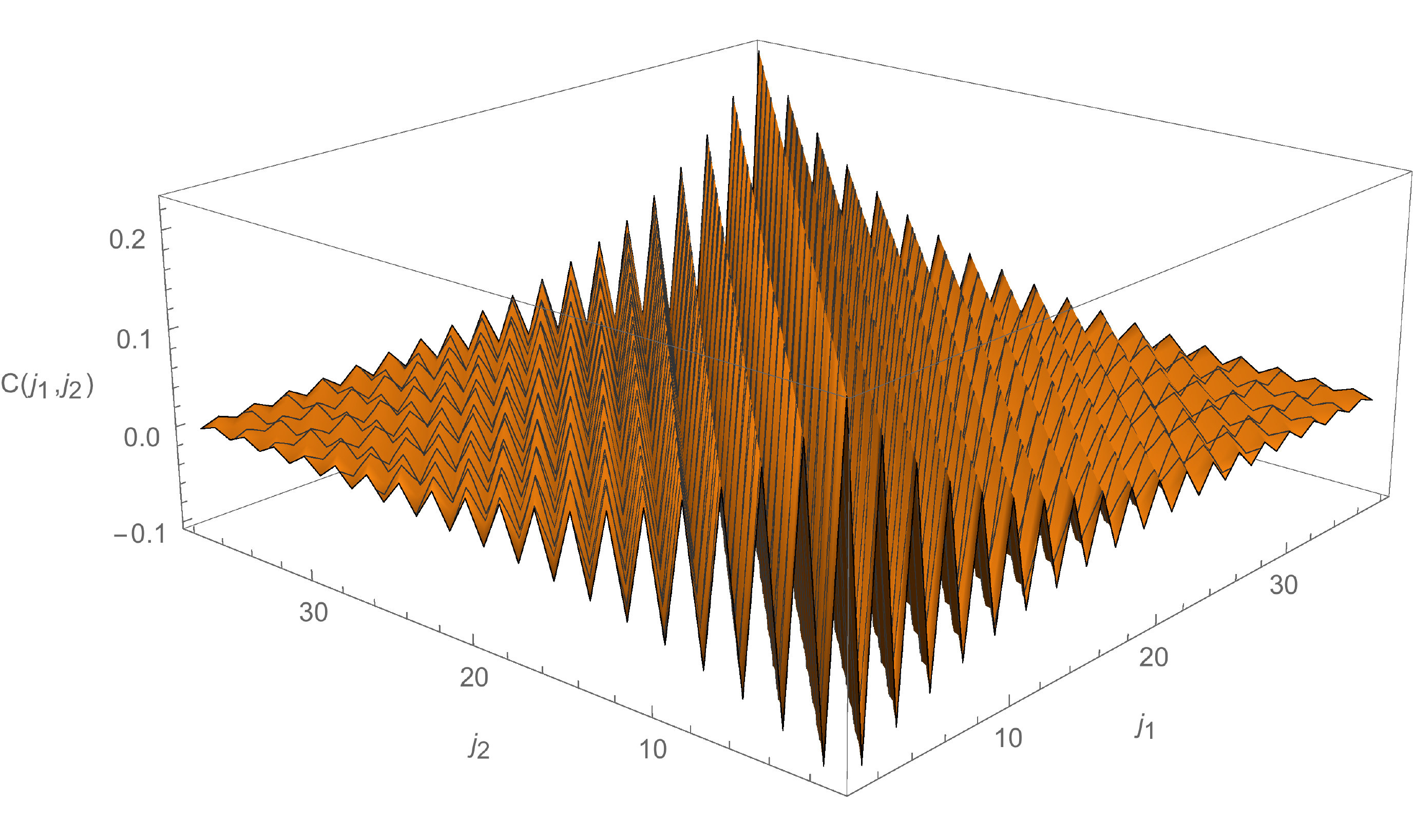}
\vspace{-1mm}
\caption{Exact connected 2-point density-density correlation function in NESS for a system size $n=40$ and driving parameters $\lambda=9/10,\mu=19/20$. Note that the boundary frame $j_{1,2}\in\{1,n\}$ is excluded from the plot for clarity.}
\label{2point}
\end{figure}

Similarly one computes a steady-state 2-point density-density correlation function, say for even-even sites $2k$,$2k'$, assuming $k < k'$:
\bea
\tilde{C}_{2k,2k'}&=&\frac{1}{Z_n}\sum_{s_1,s_2,\dots s_n} s_{2k} s_{2k'}\,p_{s_1,s_2,\dots s_n} = 
\frac{1}{Z_n} \un{l}\cdot T^{k-1} D_{\rm e} T^{k'-k}D_{\rm e} T^{m-1-k'} \un{r} \nonumber \\
&=& \frac{D_{\rm e}\un{l}\cdot T^{k'-k} D_{\rm e}\un{r}}{\un{l}\cdot T^{k'-k} \un{r}}, 
\label{2pointC}
\eea
and similarly for other pairs of sites, even-odd, odd-even and odd-odd, exchanging the corresponding $D_{\rm e}$ with $D_{\rm o}$.
Then, the connected 2-point function is defined as
\be
C_{j,j'} = \tilde{C}_{j,j'} - \rho_{j} \rho_{j'}.
\label{conn}
\ee
Writing an eigenvalue decomposition of the transfer matrix as
\be
T = \sum_{\nu=1}^3 \tau_\nu\, \un{\psi}_\nu \otimes \un{\phi}_\nu
\ee
where $\un{\phi}_1 = \un{l}/\sqrt{\un{l}\cdot\un{r}}$, $\un{\psi}_1 = \un{r}/\sqrt{\un{l}\cdot\un{r}}$, so that $\phi_\nu \cdot \psi_{\nu'} = \delta_{\nu,\nu'}$,
we can rewrite the connected correlation explicitly as
\be
C_{2k,2k'} = \frac{(D_{\rm e}\un{l}\cdot \un{\psi}_2)(\un{\phi}_2 \cdot D_{\rm e} \un{r})}{\un{l}\cdot\un{r}} \left(\frac{\tau_2}{\tau_1}\right)^{k'-k}
+  \frac{(D_{\rm e}\un{l}\cdot \un{\psi}_3)(\un{\phi}_3 \cdot D_{\rm e} \un{r})}{\un{l}\cdot\un{r}} \left(\frac{\tau_3}{\tau_1}\right)^{k'-k}\,.\quad
\ee
This demonstrates that the connected 2-point correlation function (\ref{conn}) in the bulk $1 < j,j' < n$ is indeed only a function of the difference of indices (positions) which decays exponentially $\sim \exp(-|j-j'|/\ell)$ with the correlation scale
$\ell = 1/\log |\tau_1/\tau_2|$, and depends {\em only} on the difference driving parameters $\lambda,\mu$ and {\em not} on $\alpha,\beta,\gamma,\delta$ separately. See Fig.~\ref{2point} for an example. Extension of such calculations to higher $k$-point connected 
correlations is straightforward; they would all decay exponentially $\sim \exp(-|j-j'|/\ell)$ in difference of any pair of adjacent spatial coordinates $j,j'$.

And finally, let us compute the steady state soliton currents. The current can be computed as the density of right-movers minus the density of left movers \cite{bob}.
The density of right-movers, computed as 
\be J_{\rm R} = Z_n^{-1} \sum_{\un{s}} s_{2k} s_{2k+1} p_{\un{s}} = Z_n^{-1} \un{l}\cdot T^{k-1} D_{\rm e} D_{\rm o} T^{m-1-k} \un{r} = 
\frac{\un{l}\cdot D_{\rm e}D_{\rm o}\un{r}}{\un{l}\cdot \un{r}},
\ee 
is independent of location $k$ in the steady state and reads explicitly
\be
J_{\rm R} = \frac{\lambda + 2}{ \lambda +\mu +8 - \lambda \mu}.
\ee
Similarly, the density of left-movers reads
\bea
J_{\rm L} &=& Z_n^{-1} \sum_{\un{s}} s_{2k+1} s_{2k+2} p_{\un{s}} =  \un{l}\cdot T^{k-1} D_{\rm o} T D_{\rm e} T^{m-1-k} \un{r} = 
\frac{\un{l}\cdot D_{\rm o}TD_{\rm e}\un{r}}{\tau_1\un{l}\cdot \un{r}}
\nonumber \\
&=&  \frac{\mu + 2}{ \lambda +\mu +8 - \lambda \mu},
\eea
with the overall steady state soliton current
\be
J = J_{\rm R}-J_{\rm L} =  \frac{\lambda - \mu}{ \lambda +\mu +8 - \lambda \mu}.
\label{solcur}
\ee
Note the expected linear-response behaviour for small driving parameters, namely the current becomes linearly proportional 
to the bias $J\sim \frac{1}{8} (\lambda - \mu)$.

\section{Local conservation laws}

\label{sect:cl}
 
Let us now try to approach the problem of finding possibly a complete set of independent conservation laws of RCA54 dynamics from a more formal nonequilibrium point of view. We shall take an approach analogous to the construction of quasilocal conservation laws of integrable quantum spin chains via dissipative boundary driving \cite{prl,review}. First, let us equip the space ${\cal S}$ of probability distributions with a Hilbert inner product
\be
(\un{u}|\un{v}) = 2^{-n}\,\un{u}\cdot\un{v} =  2^{-n} \sum_{\un{s}\in {\cal C}} u_{\un{s}} v_{\un{s}}, \quad \un{u},\un{v}\in{\cal S}.
\label{metric}
\ee
Writing orthogonal basis vectors on $\RR^2$ as $\un{\omega}^0 = (1,1)$ and $\un{\omega}^1 = (1,-1)$, we can write a convenient {\em orthonormal basis} $\{\un{\omega}^{b_1,b_2,\ldots,b_n},b_j\in\ZZ_2 \}$ of ${\cal S}$ as
\be
\un{\omega}^{\un{b}} = \un{\omega}^{b_1}\otimes \un{\omega}^{b_2}\otimes \cdots \otimes \un{\omega}^{b_n}, \qquad (\un{\omega}^{\un{b}}|\un{\omega}^{\un{b}'}) = \delta_{\un{b},\un{b}'},
\label{MPAomega}
\ee
meaning that 
\be
\un{\omega}^\un{b}_\un{s} = \prod_{j=1}^n \omega^{b_j}_{s_j} = (-1)^{\sum_{j=1}^n b_j s_j}.
\label{defomega}
\ee
The PSA  (\ref{eq:patch}) for the unnormalized NESS then immediately translates to a more standard form of a {\em matrix product ansatz}
\be
\un{p} = \sum_{\un{b}\in\{0,1\}^n} (\un{l}^{b_1}\cdot\sigma^{b_2}_1 \sigma^{b_3}_2 T \sigma^{b_4}_1 \sigma^{b_5}_2 T \cdots T \sigma^{b_{n-2}}_1 \sigma^{b_{n-1}}_2 \un{r}^{b_n}) \un{\omega}^{\un{b}},
\label{pness}
\ee
and analogous expression for $\un{p}'$, with {\em contravariant} boundary vectors
\be
\un{l}^b = \half \sum_{s\in\{0,1\}} (-1)^{b s} \un{r}_s,\quad \un{r}^b = \half \sum_{s\in\{0,1\}} (-1)^{b s} \un{r}_s,
\label{defcontra}
\ee
and $\sigma^b_{1,2}$, $b\in\{0,1\}$ are diagonal $4\times 4$ matrices, defined as
\be
\sigma^b_1 = \sigma^b \otimes \one_2,\quad \sigma^b_2 = \one_2 \otimes \sigma^b,\quad  \sigma^b = \half\begin{pmatrix} 1 & 0 \cr 0 & (-1)^b \end{pmatrix}.
\label{defsigma}
\ee
Clearly, it can be directly verified that with our definitions (\ref{defomega},\ref{defcontra},\ref{defsigma}), the expression (\ref{pness}) is equivalent to the first line of Eqs.~(\ref{eq:patch}),
since 
\be 
\sum_b \un{l}^b \omega^b_s = \un{l}_s, \quad \sum_b \un{r}^b \omega^b_s = \un{r}_s,
\quad
\sum_b \sigma^b \omega^b_s = \begin{pmatrix} \delta_{s,0} & 0 \cr 0 & \delta_{s,1}\end{pmatrix}.
\ee
A special state ${(\un{\omega}^0)}^{\otimes n}=\un{\omega}^{00\ldots 0}$ represents a uniform distribution over all $2^n$ configurations (`infinite temperature state').
Any state of the form
$ 
\un{\psi}^{(k,r)} = (\un{\omega}^0)^{\otimes k} \otimes \un{v} \otimes (\un{\omega}^0)^{\otimes (n-k-r)}
$ 
for some $2^r$ dimensional vectors $\un{v}\in(\RR^{2})^{\otimes r}$ is defined as a $r$-{\em local} state, supported on sites $[k,k+r-1]$, and its components may depend {\em only} on the coordinates from the supported set $s_{k},\ldots,s_{k+r-1}$,
namely $\psi^{(k,r)}_{s_1,s_2,\ldots,s_n} = v_{s_k,s_{k+1},\ldots,s_{k+r-1}}$.
In fact, introduction of the Hilbert space metric (\ref{metric}) identifies the state space with its dual --- the space of observables, so one may interpret a vector $\un{\psi}^{(k,r)}$ also as a $r$-local observable.
For example, $\un{\rho}^{(j)} = (\un{\omega}^0)^{\otimes (j-1)} \otimes (0,1) \otimes (\un{\omega}^0)^{\otimes (n-j)}$ is the {\em density}, with expectation  (\ref{density}) given  as 
\be
\rho_j = (\un{\rho}^{(j)}|\un{p}).
\ee

\begin{definition}
A local conservation law $\un{Q} \in {\cal S}$ of a boundary driven RCA is defined as  an extensive sum of a shifted $r-$local observable (for some {\em even} integer $r$ independent of size $n$), written in terms of a vector
$\un{q} \in \RR^{2^r}$, 
\be
\un{Q} = \sum_{k=1}^{(n-r)/2} (\un{\omega}^0)^{\otimes(2k-1)} \otimes \un{q} \otimes (\un{\omega}^0)^{\otimes(n-r-2k+1)} 
\ee
for which its time-difference in one step is localized near the boundaries of the system. More precisely,
\be
U \un{Q} - \un{Q} = \un{g} \otimes (\un{\omega}^0)^{\otimes (n-r')} + (\un{\omega}^0)^{\otimes (n-r')}\otimes \un{h},
\label{conslaw}
\ee
for some {\em remainder} observables, specified by vectors $\un{g},\un{h} \in \RR^{2^{r'}}$,  localized near boundaries with $n$-independent support size $r'$.
\end{definition}
Since, when approaching the thermodynamic limit $n\to \infty$, the square norm of translationally invariant sum of local observables is {\em extensive} in $n$, 
\be
(\un{Q}|\un{Q}) = \left(\frac{\un{q}\cdot\un{q}}{2^r}\right)\left(\frac{n-r}{2}\right) \propto n,
\ee
while the remainder --- RHS of (\ref{conslaw}) has a bounded (in $n$) norm, we can conclude that such $\un{Q}$ is exactly conserved in the bulk in the thermodynamic limit. A formal proof, invoking causality of RCA in place of the Lieb-Robinson bound, would be a straightforward extension of an analogous result for quantum chains \cite{ip_cmp}.

We shall now derive exact conservation laws of CA rule 54 using exactly the same strategy as for dissipatively boundary driven quantum chains \cite{prl,review}.
The NESS probability vector $\un{p}$  is already a potential candidate for a conservation law, Eq. (\ref{conslaw}), since 
\be 
U\un{p} - \un{p}=0,
\label{eq:ss}
\ee 
provided it could be identified with a local observable. This is possible in the trivial (equlibrium) case of zero biases $\lambda=\mu=0$, and $\alpha,\gamma=0$, where NESS is trivial $\un{p}^0 = (\un{\omega}^0)^{\otimes n}$. Let us set $\alpha=\gamma=0$ for simplicity in the following, while considering arbitrary values of $\alpha,\gamma$ would only alter the boundary cells and hence the remainder terms $\un{g},\un{h}$ and not the bulk density $\un{q}$.
Taking a derivative of Eq.~(\ref{eq:ss}) with expression (\ref{pness}), with respect to, say $\lambda$, and putting $\lambda=\mu=0$ at the end, we obtain exactly a conservation law (\ref{conslaw}), by identifying:
\be U = U|_{\lambda=\mu=0},\qquad \un{Q}^\lambda = \partial_\lambda \un{p}|_{\lambda=\mu=0},
\ee
with $(r=4)$-local bulk density
\be
\un{q}^\lambda = \sum_{a,a',b,b'\in\{0,1\}} (\un{v}\cdot \sigma^{a}_1 \sigma^{a'}_2 T_\lambda \sigma^{b}_1 \sigma^{b'}_2 T_0 \un{v}) \un{\omega}^{aa'bb'},
\ee
where
\bea
&&\un{v} =  \un{l}^0|_{\lambda=\mu=0} = \un{r}^0|_{\lambda=\mu=0}=(1,1,1,1),\nonumber \\
&& T_\lambda = \partial_\lambda T|_{\lambda=\mu=0}=\left(
\begin{array}{cccc}
 0 & 0 & 0 & 0 \\
 1 & 1 & -\frac{3}{2} & -\frac{1}{2} \\
 1 & 1 & 1 & 1 \\
 \frac{3}{2} & \frac{3}{2} & 0 & 1 \\
\end{array}
\right), \; T_0 = T|_{\lambda=\mu=0}=\left(
\begin{array}{cccc}
 1 & 1 & 1 & 1 \\
 1 & 1 & 1 & 1 \\
 1 & 1 & 1 & 1 \\
 1 & 1 & 1 & 1 \\
\end{array}\right).
\nonumber
\eea
The local remainder (boundary) terms $\un{g},\un{h}$ are generated through terms where $\partial_\lambda$ either hits $E^{\alpha,\alpha-\lambda}$ in the propagator $U$, or the boundary vectors $\un{l}^{b}$, or $\un{r}^b$ in the
expression (\ref{pness}).
Similarly, we obtain another local conservation law by differentiating with respect to $\mu$ instead of $\lambda$:
\bea
&&\un{Q}^\mu = \partial_\mu \un{p}|_{\lambda=\mu=0},\\
&&\un{q}^\mu = \sum_{a,a',b,b'\in\{0,1\}}
4(\un{v}\cdot \sigma^{a}_1 \sigma^{a'}_2 T_\mu \sigma^{b}_1 \sigma^{b'}_2 \un{v})\un{\omega}^{aa'bb'},\nonumber \\
&&T_\mu = \partial_\mu T|_{\lambda=\mu=0}=
\left(
\begin{array}{cccc}
 0 & 0 & 0 & 0 \\
 1 & 1 & \frac{1}{2} & \frac{3}{2} \\
 1 & 1 & 1 & 1 \\
 -\frac{1}{2} & -\frac{1}{2} & 0 & 1 \\
\end{array}
\right).
\nonumber
\eea
Writing $\un{Q}^\pm = \half(\un{Q}^\lambda \pm \un{Q}^\mu)$, and consequently  $\un{q}^\pm = \half(\un{q}^\lambda \pm \un{q}^\mu)$, we find for the density of $\un{Q}^-$
\be
\un{q}^- = \un{\omega}^{1100}-\un{\omega}^{0110} + \un{\omega}^{0010}-\un{\omega}^{1000},
\ee
or in terms of explicit dependence on cell occupation numbers $s_j$
\be
q^-_{ss't t'} = 4 (s s' - s' t),
\ee
which is exactly (4 times) the conserved net soliton current (\ref{solcur}) as discovered in Ref.~\cite{bob}.

The second conservation law is less trivial
\bea
q^+_{ss't t'}  = 4 (s+s'+ s s' t + s' t t') - 6 (s s' + s' t),
\eea
and to best of our knowledge has not been discussed before. We note that both quantities $Q^\pm$ should be exactly conserved for a purely deterministic RCA with periodic boundary conditions.\footnote{
Since this is a $\ZZ_2$ system, two independent extensive local conservation laws $\un{Q}^\pm$ are perhaps enough for integrability. Note, however, that 
\be
Q^\pm_\un{s}=\sum_k q^\pm_{s_{2k},s_{2k+1},s_{2k+2},s_{2k+3}}
\ee 
take values in $\ZZ$ and not in 
$\ZZ_2$.}

\section{Discussion and conclusions}

\begin{figure}
 \centering	
\vspace{-1mm}
\includegraphics[width=0.65\columnwidth]{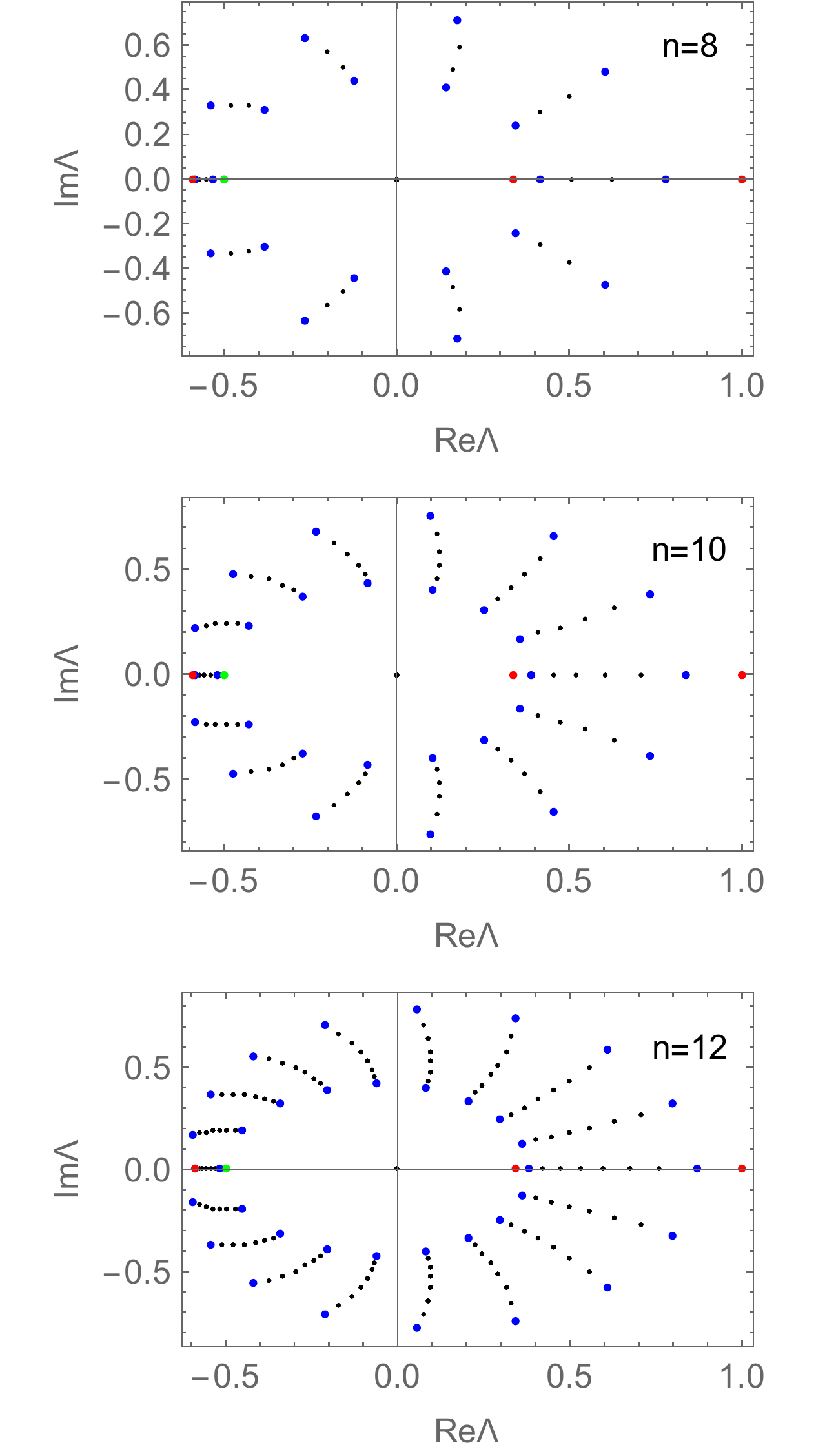}
\vspace{-1mm}
\caption{Numerical computation of full decay spectra of $2^n\times 2^n$ Markov matrices $U$ for RCA54 with $\alpha=0.9$, $\beta=0.1$, $\gamma=0.7$, $\delta=0.2$, and for three different sizes $n=8$ (top), $n=10$ (middle), $n=12$ (bottom).
The color indicates the complexity of eigenvectors as characterized by the Schmidt rank of bipartition: 2 (green), 3 (red), 6 (blue), $\ge 8$ and $n$-dependent (black). The total number of nonzero eigenvalues is $2^{n-2}$, while all the colored (low Schmidt rank) eigenvalues are {\em nondegenerate} while the other eigenvalues are typically exponentially (in $n$) degenerate. }
\label{figdecay}
\end{figure}

\label{sect:discussion}

So far we have studied only the properties of the steady state. An obvious follow up question concerns studying the full relaxation dynamics to the steady state, which amounts to studying the spectrum of decay modes, i.e. eigenvalues $\{ \Lambda_j,j=1,\ldots 2^n\}$ of the Markov matrix
\be
U \un{p}_j = \Lambda_j \un{ p}_j.
\ee
The leading eigenvalue $\Lambda_1 = 1$ corresponds to NESS, while all the others, corresponding to the so-called decay modes, lie strictly inside the unit circle 
$|\Lambda_j| < 1$, $j>1$, as following from our Theorem 1. So far we have not been able to provide any exact or rigorous results on the decay modes -- and the progress here could be very difficult as we need to devise a particular kind of Bethe ansatz -- so we instead report some intriguing results of numerical computations which should strongly motivate further study.
Apart from the Markov eigenvalues $\{ \Lambda_j\}$ we also analyzed the complexity of the corresponding eigenvectors by calculating their Schmidt rank -- the number of nonvanishing singular values of the $2^{n/2} \times 2^{n/2}$ matrix 
$(P_j)_{\un{s},\un{s}'} \equiv (p_j)_{s_1,s_2,\ldots s_{n/2},s'_1,s'_2,\ldots s'_{n/2}}$. For NESS ($j=1$) the Schmidt number is always 3 as we have proven above (e.g. it follows from the fact that the transfer matrices $T,T'$ have rank 3, see (\ref{tau})). Remarkably, there seem to be always two other eigenvectors (decay modes) with Schmidt number 3. 
The subleading eigenvector ($j=2$) corresponds to Schmidt number 6 (independent of $n$!) which gives hope that this decay mode would be analytically tractable. Curiously, there appears to be even an eigenvector of lower complexity than NESS (Schmidt number 2) with eigenvalue $\Lambda=-1/2$. The rest of the spectrum is organized in bands with end-points corresponding to non-degenerate 
eigenvectors with Schmidt number 6. Since a picture says more than a thousand words, the reader is welcome to go and stare at the Fig.~\ref{figdecay}.

In conclusion, we have presented an exact analytic treatment of the steady-state properties of a strongly interacting deterministic many-body system driven by stochastic boundaries. We considered arguably the simplest possible strongly interacting bulk dynamics which possesses certain features of integrability like solitons, namely the reversible $\ZZ_2$ cellular automaton with global conservation laws. 
To facilitate our analysis we have developed a novel algebraic  ansatz for describing strongly correlated classical many-body probability states, namely the patch state ansatz. We expect that our approach should be applicable
for constructing nonequilibrium steady-states of general classical deterministic integrable interacting theories \cite{faddeev} driven by compatible Markovian stochastic boundaries. The fundamental relation proposed here, which needs to be generalized to other integrable models, is a particular `fusion', or composition formula (\ref{ext1},\ref{ext2}) which we propose to call the $X$-system.
In an integrable lattice model with a continuous dynamical variable\footnote{For example, one may consider the Hirota equation which yields several interesting physical models in various continuum limits, e.g. the sine-Gordon model.} $\varphi_x$ at each physical site $x$, the patch tensor $X$ 
would be a function of four variables $X(\varphi_{x},\varphi_{x+1},\varphi_{x+2},\varphi_{x+3})$ and  (\ref{ext1},\ref{ext2}) would result in  some exactly solvable nonlinear functional equations.
For instance, intriguing numerical results in the integrable lattice Landau Lifshitz classical spin chain model 
suggest  \cite{bojan}  existence of a
nontrivial nonequilibrium phase transition from ballistic to diffusive steady-state and a nontrivial quasilocal conservation law in the ballistic regime, which, according to the results presented here, might be analytically treatable with appropriate integrable Markovian boundary baths.

For general classical integrable systems which are canonically defined via the zero-curvature condition in terms of a Lax pair, one should explore the possibility of connection between the patch tensors $L,X,R$ and the Lax operators generating equations of motion in the bulk. It should be noted, however, that compatibility/integrability condition between the deterministic bulk and stochastic boundaries would require all the patch tensors to explicitly depend on the Markov rates at the boundaries, just like in the model solved here. 
The question whether this can be accommodated for in terms of a single family of Lax operators, in a similar way as encoding the boundary dissipation in the complex auxiliary spin of the $U_q(\mathfrak{sl}_2)$ Lax operator in the case of boundary driven quantum $XXZ$ chains \cite{ilievski,review}, remains
an open problem for future research.

\section*{Acknowledgements}

T.P. thanks E. Ilievski for discussions and useful remarks on the manuscript. The work has been supported by the research grants 
P1-0044, J1-5439 and N1-0025 of Slovenian Research Agency (T.P.), and Spanish MICINN grant MTM2012-39101-C02-01 (C.M.-M.).

\end{document}